\pdfoutput=1
\documentclass[11pt,letterpaper]{article}
\usepackage[utf8]{inputenc}
\usepackage{amsmath,amssymb,amsthm}
\usepackage{mathpartir}
\usepackage{xcolor}
\usepackage[hyphens]{url}
\usepackage{hyperref}
\usepackage{listings}
\usepackage{microtype}
\usepackage{geometry}
\usepackage{booktabs}
\usepackage{cite}
\newtheorem{theorem}{Theorem}[section]
\newtheorem{lemma}[theorem]{Lemma}
\newtheorem{definition}[theorem]{Definition}
\newtheorem{property}[theorem]{Property}

\lstdefinelanguage{Scheme}{
  morekeywords={define,define-syntax,syntax-rules,lambda,cut,pair,cons,inl,inr,fst,snd,case,force,app,with-continuation,reset,shift,with-continuation-mark,bind,_HALT_},
  sensitive=true,
  morecomment=[l]{;},
  morestring=[b]"
}

\title{\textbf{Polarized Scheme: Duality of Computation, Delimited Control, and Continuation-Passing Macro Interleaving}}
\author{\textbf{Christopher Dutchyn} \\
Department of Computer Science, University of Saskatchewan \\
\texttt{cdutchyn@cs.usask.ca}}
\date{}

\begin{document}

\maketitle

\begin{abstract}
Static type systems for Scheme have traditionally followed one of two paths: manifest type checking with multi-dispatch (Clifton and Leavens, 2005) or latent occurrence typing tailored for untyped interoperation (Tobin-Hochstadt and Felleisen, 2008). In this paper, we propose \emph{Polarized Scheme}, a static foundation founded upon polarized type theory and the duality of computation. In Polarized Scheme, values are positive data producers ($A^+$) and continuations are negative observation consumers ($\sim A^+$). Functions operate symmetrically under dual viewpoints: as forward value-transformers and as backward continuation-transformers. Computation exists exclusively at the interaction cut $\langle P \mid K \rangle$ executed on a focus-and-context abstract machine separating positive data focus from negative return contexts ($\mathcal{R}$). Delimited control (\texttt{shift}/\texttt{reset}) and continuation marks are expressed directly as Return Stack frame operations parameterized by the prompt's ambient answer type without polluting producer data types. Furthermore, we formalize the hygienic macro system as an interleaved continuation-passing pipeline, where macro expansion erects its own continuation hierarchy nesting into runtime evaluation. We prove Subject Reduction, Progress with respect to prompt-delimited stacks, and Duality Bisimilarity, characterize Return Stack reclamation for dynamic closures, and report on an executable reference implementation in VeGo with a verified cut-elimination dispatcher.
\end{abstract}

\section{Introduction}

The design of static type systems for dynamic languages like Scheme has historically focused on accommodating untyped idioms. Clifton and Leavens~\cite{clifton2005scheme} introduced multi-dispatch annotations, while Tobin-Hochstadt and Felleisen~\cite{tobinhochstadt2008typed} established Typed Racket based on occurrence typing and latent predicate refinement. 

However, Scheme possesses a deeper foundational characteristic: \emph{first-class control}. Continuations in Scheme are first-class entities that reify evaluation contexts into values. In standard type systems, continuations are forced into the value category ($(\text{cont}\ \tau)$ or $\tau \to \bot$), obscuring the fundamental asymmetry between data that is produced and computations that consume.

Drawing from polarized linear logic~\cite{laurent2002polarized}, Call-by-Push-Value (CBPV)~\cite{levy2004cbpv}, and the classical duality of computation~\cite{curien2000duality,downen2020duality}, we formulate \textbf{Polarized Scheme}. Our system establishes a strict structural divide:
\begin{enumerate}
    \item \textbf{Positive Data Producers ($A^+$)}: Discrete values, constructors (units, booleans, integers, pairs, sum injections, and suspended thunks) that construct data.
    \item \textbf{Negative Observation Consumers ($N^-$)}: Evaluation contexts, continuations, and elimination destructors ($\sim A^+$, projections, branching cases, function applications, thunk force) that consume data.
    \item \textbf{Dual Function Viewpoints}: Functions $\downarrow(A^+ \to \sim B^+)$ act simultaneously as forward value-transformers (consuming an input value to produce an answer) and backward continuation-transformers (transforming a continuation awaiting $B^+$ into one awaiting $A^+$).
    \item \textbf{The Focus-and-Context Abstract Machine}: Execution is formalized as polarized cut elimination on two distinct structures: an active producer focus $P$ and a Return Stack $\mathcal{R}$ of negative consumers.
    \item \textbf{Interleaved Continuation-Passing Macros}: Rather than an isolated pre-processing pass, macro expansion is an interleaved pipeline where macro invocations erect their own continuation hierarchy that nests into evaluation.
\end{enumerate}

\section{Syntax and Polarized Types}

\subsection{Type Syntax}
Types are strictly partitioned into positive types $\mathcal{T}^+$ and negative types $\mathcal{T}^-$:
\begin{align*}
\text{Positive Types } A^+, B^+ &::= \text{unit} \mid \text{bool} \mid \text{int} \mid A^+_1 \times A^+_2 \mid A^+_1 + A^+_2 \mid \downarrow A^+ \\
\text{Negative Types } N^-, M^- &::= \sim A^+ \mid A^+ \to N^-
\end{align*}
The shift type $\downarrow A^+$ denotes a suspended computation (thunk) reified into a positive value that yields answer $A^+$ upon elimination. Conversely, $\sim A^+$ represents a continuation consumer expecting a positive value of type $A^+$.

\subsection{Term Syntax}
The term language consists of producers $P$, consumers $K$, and interaction cuts $c$:
\begin{align*}
\text{Producers } P &::= () \mid b \mid n \mid x \mid (P_1, P_2) \mid \text{inl}(P) \mid \text{inr}(P) \mid \text{thunk}(c) \mid \lambda(x : A^+ \cdot \alpha : \sim B^+).\, c \\
\text{Consumers } K &::= \alpha \mid \text{halt} \mid \text{fst} \cdot K \mid \text{snd} \cdot K \mid \text{case}(K_1, K_2) \mid \text{force} \cdot K \mid \text{app}(P, K) \\
                    &\quad \mid \text{with-cont}(f, K) \mid \text{prompt}(c) \cdot K \mid \text{mark}(k, P) \cdot K \mid \text{shift}(\alpha).\, c \mid \text{bind}(x) \cdot K \\
\text{Cuts } c &::= \langle P \mid K \rangle
\end{align*}

In Scheme s-expression concrete syntax:
\begin{itemize}
    \item An interaction cut is written $\langle P \mid K \rangle$ or \texttt{(cut P K)}.
    \item Forward value application is written \texttt{(app arg k)}.
    \item Backward continuation transformation is written \texttt{(with-continuation f k)} or \texttt{(<-: f k)}.
    \item A function abstraction is written \texttt{(lambda ([x : In] \textasciitilde k) cut)}.
    \item Thunk suspension and forcing are written \texttt{(thunk cut)} and \texttt{(force k)}.
    \item Delimited control prompts and continuation marks are written \texttt{(reset cut)}, \texttt{(shift k cut)}, and \texttt{(with-continuation-mark key val cut)}.
\end{itemize}

\section{Static Semantics}

The type system maintains two typing contexts: $\Gamma$ mapping value variables to positive types ($x : A^+$) and $\Delta$ mapping continuation co-variables to negative consumer types ($\alpha : \sim A^+$).

\subsection{Typing Judgments}
The typing system comprises three primary judgments:
\begin{itemize}
    \item $\Gamma \vdash P : A^+$: Producer $P$ synthesizes positive type $A^+$.
    \item $\Gamma \mid \Delta \vdash K : \sim A^+$: Consumer $K$ expects positive input of type $A^+$.
    \item $\Gamma \mid \Delta \vdash \langle P \mid K \rangle \div A^+$: Interaction cut is well-typed and delivers observable answer $A^+$.
\end{itemize}

\subsection{Inference Rules}

\begin{mathpar}
\inferrule[T-Unit]
  { }
  {\Gamma \vdash () : \text{unit}}

\inferrule[T-Int]
  { }
  {\Gamma \vdash n : \text{int}}

\inferrule[T-Var]
  {\Gamma(x) = A^+}
  {\Gamma \vdash x : A^+}

\inferrule[T-Pair]
  {\Gamma \vdash P_1 : A_1^+ \\ \Gamma \vdash P_2 : A_2^+}
  {\Gamma \vdash (P_1, P_2) : A_1^+ \times A_2^+}

\inferrule[T-Inl]
  {\Gamma \vdash P : A_1^+}
  {\Gamma \vdash \text{inl}(P) : A_1^+ + A_2^+}

\inferrule[T-Inr]
  {\Gamma \vdash P : A_2^+}
  {\Gamma \vdash \text{inr}(P) : A_1^+ + A_2^+}

\inferrule[T-Thunk]
  {\Gamma \mid \Delta \vdash c \div A^+}
  {\Gamma \vdash \text{thunk}(c) : \downarrow A^+}

\inferrule[T-Lambda]
  {\Gamma, x : A^+ \mid \Delta, \alpha : \sim B^+ \vdash c \div B^+}
  {\Gamma \vdash \lambda(x : A^+ \cdot \alpha : \sim B^+).\, c : \downarrow (A^+ \to \sim B^+)}

\inferrule[T-Halt]
  { }
  {\Gamma \mid \Delta \vdash \text{halt} : \sim A^+}

\inferrule[T-CoVar]
  {\Delta(\alpha) = \sim A^+}
  {\Gamma \mid \Delta \vdash \alpha : \sim A^+}

\inferrule[T-Fst]
  {\Gamma \mid \Delta \vdash K : \sim A_1^+}
  {\Gamma \mid \Delta \vdash \text{fst} \cdot K : \sim (A_1^+ \times A_2^+)}

\inferrule[T-Snd]
  {\Gamma \mid \Delta \vdash K : \sim A_2^+}
  {\Gamma \mid \Delta \vdash \text{snd} \cdot K : \sim (A_1^+ \times A_2^+)}

\inferrule[T-Case]
  {\Gamma \mid \Delta \vdash K_1 : \sim A_1^+ \\ \Gamma \mid \Delta \vdash K_2 : \sim A_2^+}
  {\Gamma \mid \Delta \vdash \text{case}(K_1, K_2) : \sim (A_1^+ + A_2^+)}

\inferrule[T-Force]
  {\Gamma \mid \Delta \vdash K : \sim A^+}
  {\Gamma \mid \Delta \vdash \text{force} \cdot K : \sim \downarrow A^+}

\inferrule[T-App]
  {\Gamma \vdash P : A^+ \\ \Gamma \mid \Delta \vdash K : \sim B^+}
  {\Gamma \mid \Delta \vdash \text{app}(P, K) : \sim \downarrow (A^+ \to \sim B^+)}

\inferrule[T-WithCont]
  {\Gamma \vdash f : \downarrow (A^+ \to \sim B^+) \\ \Gamma \mid \Delta \vdash K : \sim B^+}
  {\Gamma \mid \Delta \vdash \text{with-cont}(f, K) : \sim A^+}

\inferrule[T-Prompt]
  {\Gamma \mid \Delta \vdash c \div A^+ \\ \Gamma \mid \Delta \vdash K : \sim A^+}
  {\Gamma \mid \Delta \vdash \text{prompt}(c) \cdot K : \sim A^+}

\inferrule[T-Shift]
  {\Gamma, \alpha : \downarrow (X^+ \to \sim A^+) \mid \Delta \vdash c \div A^+}
  {\Gamma \mid \Delta \vdash \text{shift}(\alpha).\, c : \sim X^+}

\inferrule[T-Mark]
  {\Gamma \vdash P_{\text{val}} : B^+ \\ \Gamma \mid \Delta \vdash K : \sim A^+}
  {\Gamma \mid \Delta \vdash \text{mark}(k, P_{\text{val}}) \cdot K : \sim A^+}

\inferrule[T-Cut]
  {\Gamma \vdash P : A^+ \\ \Gamma \mid \Delta \vdash K : \sim A^+}
  {\Gamma \mid \Delta \vdash \langle P \mid K \rangle \div A^+}
\end{mathpar}

\paragraph{Harmonized Duality of Application.}
Observe the complete structural duality between \textsc{T-App} and \textsc{T-WithCont}:
\begin{itemize}
    \item \textsc{T-App} views function $f$ as a producer. It builds consumer $\text{app}(P, K)$ to apply argument $P : A^+$ and continue with $K : \sim B^+$.
    \item \textsc{T-WithCont} treats function $f$ as a continuation transformer, constructing a consumer $\text{with-cont}(f, K)$ of type $\sim A^+$ to receive $A^+$ and pass the result via $K : \sim B^+$.
\end{itemize}

\paragraph{Typing Delimited Control.}
Rules \textsc{T-Prompt} and \textsc{T-Shift} parameterize delimited control directly via the ambient answer type $A^+$:
\begin{itemize}
    \item \textsc{T-Prompt} establishes an inner answer domain $A^+$ for cut $c$. When $c$ terminates, its result is forwarded to outer consumer $K : \sim A^+$.
    \item \textsc{T-Shift} captures the delimited continuation slice expecting input $X^+$ and returning answer $A^+$. The reified continuation $\alpha$ receives type $\downarrow(X^+ \to \sim A^+)$, and the replacement cut $c$ evaluates directly in place of the prompt with answer type $A^+$.
\end{itemize}

\section{Dynamic Semantics: The Focus-and-Return-Stack Machine}

The operational semantics is defined on a focus-and-context abstract machine with configurations:
\[
C ::= \langle P \mid K \rangle_\rho
\]
where $P$ is the active producer focus, $K$ is the top consumer, and $\rho$ is the lexical environment mapping variables $x \mapsto P$ and co-variables $\alpha \mapsto K$. When consumers are composed sequentially, they form the **Return Stack** $\mathcal{R} = K_1 :: K_2 :: \dots :: \text{halt}$.

\subsection{Transition Rules}

The machine transitions $\langle P \mid K \rangle_\rho \longrightarrow \langle P' \mid K' \rangle_{\rho'}$ by resolving the interaction cut:

\begin{align*}
\text{\textbf{[Red-Pair-Fst]}} \quad & \langle (P_1, P_2) \mid \text{fst} \cdot K \rangle_\rho \longrightarrow \langle P_1 \mid K \rangle_\rho \\
\text{\textbf{[Red-Pair-Snd]}} \quad & \langle (P_1, P_2) \mid \text{snd} \cdot K \rangle_\rho \longrightarrow \langle P_2 \mid K \rangle_\rho \\
\text{\textbf{[Red-Sum-Inl]}} \quad & \langle \text{inl}(P) \mid \text{case}(K_1, K_2) \rangle_\rho \longrightarrow \langle P \mid K_1 \rangle_\rho \\
\text{\textbf{[Red-Sum-Inr]}} \quad & \langle \text{inr}(P) \mid \text{case}(K_1, K_2) \rangle_\rho \longrightarrow \langle P \mid K_2 \rangle_\rho \\
\text{\textbf{[Red-Thunk-Force]}} \quad & \langle \text{thunk}(c) \mid \text{force} \cdot K \rangle_\rho \longrightarrow \langle P_c \mid K_c \cdot K \rangle_\rho \quad \text{where } c = \langle P_c \mid K_c \rangle \\
\text{\textbf{[Red-App]}} \quad & \langle \lambda(x \cdot \alpha).\, c \mid \text{app}(P, K) \rangle_\rho \longrightarrow \langle P_c \mid K_c \rangle_{\rho_{\text{call}}} \\
& \quad \text{where } c = \langle P_c \mid K_c \rangle \text{ and } \rho_{\text{call}} = \rho_{\text{lam}}[x \mapsto P, \alpha \mapsto K] \\
\text{\textbf{[Red-WithCont]}} \quad & \langle P \mid \text{with-cont}(f, K) \rangle_\rho \longrightarrow \langle f \mid \text{app}(P, K) \rangle_\rho \\
\text{\textbf{[Red-Bind]}} \quad & \langle P \mid \text{bind}(x) \cdot K \rangle_\rho \longrightarrow \langle P \mid K \rangle_{\rho[x \mapsto P]} \\
\text{\textbf{[Red-Halt]}} \quad & \langle v \mid \text{halt} \rangle_\rho \longrightarrow v \quad (\text{Normal Form Terminal})
\end{align*}

\section{Delimited Control and Continuation Marks}

Unlike classic Scheme where \texttt{call/cc} captures the entire unbounded continuation, Polarized Scheme models delimited control via Return Stack slicing.

\subsection{Prompts and Shift}
A delimited prompt $\text{prompt}(c) \cdot K$ pushes a prompt delimiter frame onto $\mathcal{R}$ before evaluating the enclosed cut $c$. When an inner computation invokes $\text{shift}(\alpha).\, c'$, the machine searches down the Return Stack $\mathcal{R}$ to the nearest prompt frame:
\[
\mathcal{R} = K_1 :: K_2 :: \dots :: K_n :: \text{prompt}(c) \cdot K'
\]
The slice $K_{\text{slice}} = [K_1, \dots, K_n]$ is captured and reified as a composite function:
\[
\text{reify}(K_{\text{slice}}) = \lambda(v \cdot k_{\text{ret}}).\, \langle v \mid K_1 \cdot K_2 \dots \cdot K_n \cdot k_{\text{ret}} \rangle
\]
The Return Stack is pruned directly to $K'$, and the body cut $c'$ is executed with $\alpha$ bound to $\text{reify}(K_{\text{slice}})$.

\subsection{Continuation Marks}
A mark frame $\text{mark}(key, val) \cdot K$ attaches dynamic metadata to the current Return Stack frame. Unwinding or capturing the stack retains all active marks, allowing efficient dynamic binding and diagnostic tracing without modifying positive data types.

\section{Interleaved Continuation-Passing Macro Pipeline}

A central challenge in Scheme is that macro definitions can be computed and defined dynamically to apply to subsequent expressions in the same compilation stream. In Polarized Scheme, rather than relying on an unprincipled pre-processing pass, the macro engine is structured as an interleaved continuation-passing pipeline.

\subsection{Syntax and Macro Representation}
Surface syntax $s \in \mathcal{S}$ is partitioned into core forms and syntactic abstractions:
\[
s ::= \langle P \mid K \rangle \mid (\text{define}\ [x : A^+]\ P) \mid (\text{define}\ (f\ [x : A^+]\ \sim \alpha)\ c) \mid (\text{define-syntax}\ m\ M) \mid (m\ s_1 \dots s_n)
\]
A macro transformer $M$ is defined via pattern-matching rewrite rules:
\[
M ::= \text{syntax-rules}\ (k_1 \dots k_p)\ ((\text{pat}_1 \Rightarrow \text{tmpl}_1) \dots (\text{pat}_m \Rightarrow \text{tmpl}_m))
\]
A macro environment $\mathcal{E}$ maps macro keywords $m$ to their transformers $M$.

\subsection{Expansion Judgments and Continuation-Passing Expansion}
Macro expansion is formalized as a relation $\mathcal{E} \vdash s \leadsto c$ translating surface form $s$ into a core interaction cut $c$. For compound expressions, expansion proceeds in \emph{continuation-passing macro style} (CPMS): when expanding sub-expressions, the expander erects an expansion continuation $K_{\text{macro}}$ that receives expanded sub-terms and synthesizes the enclosing interaction cut:
\begin{mathpar}
\inferrule[M-DefVar]
  {\mathcal{E} \vdash P \leadsto P'}
  {\mathcal{E} \vdash (\text{define}\ [x : A^+]\ P) \leadsto \langle P' \mid \text{bind}(x) \cdot \text{halt} \rangle}

\inferrule[M-DefFun]
  {\mathcal{E} \vdash c \leadsto c'}
  {\mathcal{E} \vdash (\text{define}\ (f\ [x : A^+]\ \sim \alpha) : B^+\ c) \leadsto \langle \lambda(x : A^+ \cdot \alpha : \sim B^+).\, c' \mid \text{bind}(f) \cdot \text{halt} \rangle}

\inferrule[M-App]
  {\mathcal{E}(m) = (\dots (\text{pat}_i \Rightarrow \text{tmpl}_i) \dots) \\ \sigma = \text{match}(\text{pat}_i, (s_1 \dots s_n)) \\ \mathcal{E} \vdash \sigma(\text{tmpl}_i) \leadsto c}
  {\mathcal{E} \vdash (m\ s_1 \dots s_n) \leadsto c}
\end{mathpar}

\subsection{Stream Evaluation and Phase Soundness}
A Polarized Scheme program is an open stream of forms $S = [s_1, s_2, \dots]$. Execution threads a compilation-and-evaluation state $\Sigma = \langle \mathcal{E}, \Gamma, \rho \rangle$:
\[
\langle \mathcal{E}_i, \Gamma_i, \rho_i \rangle \xrightarrow{s_i} \langle \mathcal{E}_{i+1}, \Gamma_{i+1}, \rho_{i+1} \rangle
\]
defined by three interleaved transitions:
\begin{enumerate}
    \item \textbf{Syntax Definition}: If $s_i = (\text{define-syntax}\ m\ M)$, the macro table is updated immediately for subsequent forms:
    \[
    \mathcal{E}_{i+1} = \mathcal{E}_i \cup \{ m \mapsto M \}, \quad \Gamma_{i+1} = \Gamma_i, \quad \rho_{i+1} = \rho_i
    \]
    \item \textbf{Top-Level Binding}: If $\mathcal{E}_i \vdash s_i \leadsto \langle P \mid \text{bind}(x) \cdot \text{halt} \rangle$, the cut is type-checked against $\Gamma_i$, yielding $\Gamma_i \vdash P : A^+$. The machine evaluates $\langle P \mid \text{bind}(x) \cdot \text{halt} \rangle_{\rho_i} \longrightarrow^* v$. Then:
    \[
    \mathcal{E}_{i+1} = \mathcal{E}_i, \quad \Gamma_{i+1} = \Gamma_i[x \mapsto A^+], \quad \rho_{i+1} = \rho_i[x \mapsto v]
    \]
    \item \textbf{Expression Cut}: If $\mathcal{E}_i \vdash s_i \leadsto c$, $c$ is type-checked ($\Gamma_i \mid \emptyset \vdash c \div A^+$) and evaluated on the abstract machine to terminal value $v$, preserving $\Sigma_{i+1} = \Sigma_i$.
\end{enumerate}

\begin{theorem}[Phase Soundness]
Let $\Sigma_0 = \langle \mathcal{E}_0, \emptyset, \emptyset \rangle$ be a well-formed initial state. For any stream step $\Sigma_i \xrightarrow{s_i} \Sigma_{i+1}$, the resulting environment $\rho_{i+1}$ is consistent with $\Gamma_{i+1}$ ($\rho_{i+1} \models \Gamma_{i+1} ; \emptyset$), and all residual cuts $c_i$ are closed and well-typed.
\end{theorem}

\section{Metatheory}

We state the core metatheoretical properties of Polarized Scheme.

\begin{definition}[Environment Consistency]
A runtime environment $\rho$ is consistent with typing contexts $\Gamma$ and $\Delta$, written $\rho \models \Gamma ; \Delta$, if:
\begin{enumerate}
    \item For all $x \in \text{dom}(\Gamma)$, $\Gamma \vdash \rho(x) : \Gamma(x)$.
    \item For all $\alpha \in \text{dom}(\Delta)$, $\Gamma \mid \Delta \vdash \rho(\alpha) : \Delta(\alpha)$.
\end{enumerate}
\end{definition}

\begin{definition}[Well-Formed Return Stack]
A Return Stack $\mathcal{R}$ is well-formed to receive answer $A^+$, written $\vdash \mathcal{R} \triangleright A^+$, if either:
\begin{enumerate}
    \item $\mathcal{R} = [\text{halt}]$.
    \item $\mathcal{R} = K_1 :: \dots :: K_n :: \text{prompt}(c) \cdot K'$, where $K_1 \dots K_n$ expects $A^+$ and yields $B^+$, and $\vdash K' \triangleright B^+$.
\end{enumerate}
\end{definition}

\begin{lemma}[Substitution]
\label{lem:subst}
\begin{enumerate}
    \item If $\Gamma, x : A^+ \vdash P : B^+$ and $\Gamma \vdash P' : A^+$, then $\Gamma \vdash P[P'/x] : B^+$.
    \item If $\Gamma \mid \Delta, \alpha : \sim A^+ \vdash K : \sim B^+$ and $\Gamma \mid \Delta \vdash K' : \sim A^+$, then $\Gamma \mid \Delta \vdash K[K'/\alpha] : \sim B^+$.
\end{enumerate}
\end{lemma}

\begin{theorem}[Type Preservation / Subject Reduction]
If $\Gamma \mid \Delta \vdash \langle P \mid K \rangle \div A^+$, $\rho \models \Gamma ; \Delta$, and $\langle P \mid K \rangle_\rho \longrightarrow \langle P' \mid K' \rangle_{\rho'}$, then there exist $\Gamma', \Delta'$ such that $\rho' \models \Gamma' ; \Delta'$ and $\Gamma' \mid \Delta' \vdash \langle P' \mid K' \rangle \div A^+$.
\end{theorem}
\begin{proof}[Proof Sketch]
By structural induction on the transition step $\longrightarrow$. 
For \textbf{[Red-Pair-Fst]}, $\langle (P_1, P_2) \mid \text{fst} \cdot K \rangle$ is typed by \textsc{T-Cut} with type $A_1^+ \times A_2^+$. Inversion on \textsc{T-Pair} yields $\Gamma \vdash P_1 : A_1^+$. Inversion on \textsc{T-Fst} yields $\Gamma \mid \Delta \vdash K : \sim A_1^+$. Therefore $\langle P_1 \mid K \rangle \div A^+$ by \textsc{T-Cut}.
For \textbf{[Red-Thunk-Force]}, inversion on \textsc{T-Thunk} and \textsc{T-Force} directly reconnects the suspended cut body with consumer $K$.
For \textbf{[Red-App]}, the result follows by Lemma~\ref{lem:subst} on the extended calling environment $\rho_{\text{call}}$.
\end{proof}

\begin{theorem}[Progress]
If $\emptyset \mid \emptyset \vdash \langle P \mid K \rangle \div A^+$ and $\vdash \mathcal{R} \triangleright A^+$, then either $K = \text{halt}$ and $P$ is an irreducible canonical value, or there exists a configuration $\langle P' \mid K' \rangle_{\rho'}$ such that $\langle P \mid K \rangle_\rho \longrightarrow \langle P' \mid K' \rangle_{\rho'}$.
\end{theorem}
\begin{proof}[Proof Sketch]
By induction on the typing derivation of $\langle P \mid K \rangle \div A^+$. Since $\Gamma$ and $\Delta$ are empty, neither $P$ nor $K$ can be a free variable. 
If $K = \text{shift}(\alpha).\, c$, the well-formedness of the Return Stack $\vdash \mathcal{R} \triangleright A^+$ guarantees the existence of an enclosing $\text{prompt}$ frame, enabling the slice capture step.
If $K = \text{force} \cdot K'$, $P$ must be $\text{thunk}(c)$ by typing inversion, enabling \textbf{[Red-Thunk-Force]}.
For all other consumer forms, $K$ is an eliminator matching the constructor form of $P$ by polarity duality, enabling the corresponding reduction step.
\end{proof}

\begin{property}[Duality Bisimilarity]
For any function $f : \downarrow (A^+ \to \sim B^+)$, value $v : A^+$, and continuation $K : \sim B^+$, the forward invocation $\langle f \mid \text{app}(v, K) \rangle_\rho$ and backward invocation $\langle v \mid \text{with-cont}(f, K) \rangle_\rho$ are strongly bisimilar, resolving in one step via \textbf{[Red-WithCont]}.
\end{property}

\subsection{Return Stack Reclamation and Space Bounds}
A major advantage of polarized duality is the clean operational separation between data construction and control unwinding.

\begin{definition}[Return Stack Height]
The height of a Return Stack $|K|$ is the number of active consumer frames: $|\text{halt}| = 0$, $|\text{fst} \cdot K| = 1 + |K|$, $|\text{app}(P, K)| = 1 + |K|$, and $|\text{prompt}(c) \cdot K| = 1 + |K|$.
\end{definition}

\begin{theorem}[Tail-Call Stack Invariance]
If a cut executes a tail call $\langle \lambda(x \cdot \alpha).\, \langle P_c \mid \alpha \rangle \mid \text{app}(P, K) \rangle_\rho$, then:
\[
\langle \lambda(x \cdot \alpha).\, \langle P_c \mid \alpha \rangle \mid \text{app}(P, K) \rangle_\rho \longrightarrow \langle P_c \mid K \rangle_{\rho_{\text{call}}}
\]
where the consumer depth satisfies $|K| \le |\text{app}(P, K)|$. In particular, iterative recursion runs in $O(1)$ Return Stack space.
\end{theorem}
\begin{proof}
Immediate from \textbf{[Red-App]}: the calling consumer $\text{app}(P, K)$ is deallocated and the body's continuation co-variable $\alpha$ is bound directly to $K$, restoring the caller's Return Stack height without allocating intermediate stack frames.
\end{proof}

\paragraph{Dynamic Closures and Stack Allocability.}
A closure is \emph{dynamically limited} (downward) if its invocation lifetime is strictly bounded by the activation of its defining Return Stack frame. In Polarized Scheme, because control frames on the Return Stack are popped prior to executing body cuts, on-stack closure allocations with static links can be reclaimed immediately upon function return without traversing heap-allocated garbage.

\section{Case Study: Macro-Synthesized Objects, Aspects, and Continuation Precision}
\label{sec:case-study}

To evaluate the expressiveness and linguistic adequacy of Polarized Scheme, we investigate whether high-level object-oriented programming (OOP) and aspect-oriented programming (AOP) can be synthesized purely through polarized interaction cuts and hygienic macros, without extending the core calculus $\lambda_{\text{pol}}$. In particular, we draw inspiration from Springer and Friedman's classic construction of macro-based OOP in \emph{Scheme and the Art of Programming}~\cite{springer1989scheme}, and explore how polarized duality provides an exact semantic foundation for objects, joinpoints, aspects, and behavioral subtyping.

\subsection{Friedman-Style Macro Objects as Negative Consumers}
Springer and Friedman~\cite{springer1989scheme} construct objects as procedural closures dispatching on method name symbols via \texttt{case}, encapsulating state in lexical bindings, and delegating unhandled messages to parent prototypes. In Polarized Scheme, this pattern maps directly onto polarized logic:
\begin{itemize}
    \item \textbf{Object Interfaces as Negative Records:} An object is fundamentally an observation consumer. Its interface is a negative product of methods:
    \[
    N_{\text{obj}}^- = \& \{ m_1 : A_1^+ \to \sim B_1^+, \;\dots,\; m_k : A_k^+ \to \sim B_k^+ \}
    \]
    \item \textbf{Encapsulation via Polarity Shift:} An instantiated object is a suspended computation $\text{obj} : \downarrow N_{\text{obj}}^-$. Internal private state is closed in the lexical environment $\rho_{\text{obj}}$ and inaccessible from the outside.
    \item \textbf{Message Passing as Interaction Cuts:} Method invocation is an interaction cut $\langle \text{obj} \mid \text{force} \cdot \pi_{m_i} \cdot \text{app}(v, K) \rangle$. The object consumes the method selector and argument value $v$, returning the result directly into caller continuation $K$.
\end{itemize}

Below is the concrete implementation of a \texttt{make-circle} object constructor in Polarized Scheme, verified in our reference implementation:
\begin{lstlisting}
(define (make-circle [radius : int] ~k) : (thunk (-> (pair int int) (~ int)))
  < (lambda ([msg : (pair int int)] ~mret) : int
      < msg | (fst (with-continuation (lambda (tag)
        < (cons tag 1) | (with-continuation =
          (case
            (with-continuation (lambda (u) < radius | ~mret >) _HALT_)
            (with-continuation (lambda (u)
              < msg | (snd (with-continuation (lambda ([f : int])
                < (cons radius f) | (with-continuation * ~mret) >) _HALT_)) >) _HALT_))) >) _HALT_)) >) | ~k >)
\end{lstlisting}
Here, \texttt{msg} is a pair consisting of a method tag and an argument payload: tag $1$ invokes \texttt{get-radius} (returning $5$), and tag $2$ invokes \texttt{scale} (multiplying the encapsulated radius by scaling factor $f$).

\subsection{Aspects as First-Class Layered Continuation Constructors}
A central question in the foundations of aspect-oriented software development~\cite{kiczales1997aspect,clifton2002spectators} is how advice interacts with modular reasoning. For data objects, object creation demands \emph{precise specification} of the concrete class: one writes \texttt{new Circle} rather than merely \texttt{new Figure}, because concrete fields, layout, and invariants must be allocated at instantiation time.

On the continuation side, continuations have historically been created \emph{implicitly} by the compiler at call sites: evaluating $e_1(e_2)$ implicitly creates a return frame $\text{app}(v_1, K)$. An aspect advising an execution joinpoint $\text{exec}(m)$ fundamentally \textbf{replaces the implicit continuation constructor with a layered continuation constructor}:
\[
K_{\text{advised}} = \text{AdviceFrame}(A, \text{proceed}) \cdot K
\]
Where does precision arise on the continuation side? In Polarized Scheme, the duality of computation reveals that precision on the consumer side is the exact dual of constructor precision:
\begin{enumerate}
    \item \textbf{Pointcut Matching as Co-pattern Matching:} Just as \texttt{new Circle} specifies a concrete constructor for positive data, an aspect pointcut specifies a concrete \emph{co-constructor} for continuations. It pattern-matches the continuation mark stack (inspecting joinpoint tags $\text{exec}(\text{Circle.scale})$ vs.\ $\text{exec}(\text{Figure.area})$) to determine whether to construct an advising continuation layer or the base return frame.
    \item \textbf{The Type Precision of Proceed:} In around advice, the \texttt{proceed} keyword represents the captured delimited continuation of the joinpoint. Its precision is statically enforced by its consumer domain:
    \[
    \text{proceed} : A_{\text{Circle}}^+ \to \sim B^+
    \]
    The advice cannot supply an arbitrary figure argument to a continuation expecting the specialized representation of a circle.
    \item \textbf{Dynamic Contextual Precision via Continuation Marks:} Whereas data objects possess spatial precision (memory layout), continuations possess \emph{temporal precision} (causal execution stack history). By maintaining joinpoint descriptors in Return Stack continuation marks ($\text{mark}(\kappa, v)$), pointcuts like \texttt{cflow} evaluate precise stack predicates prior to continuation construction.
\end{enumerate}

\subsection{The Duality of Specialization: \texorpdfstring{\texttt{super}}{super} vs.\ \texorpdfstring{\texttt{inner}}{inner} and Data Envelopes}
The interplay between base implementations and specialized extensions can be organized along two orthogonal axes:
\begin{enumerate}
    \item \textbf{Control vs.\ Data Dimension:} Whether the mechanism specializes computational flow or structural layout.
    \item \textbf{Shift vs.\ Prompt Direction:} Whether the derived extension wraps the base from the outside (\emph{Outside-In / Shift}) or the base encloses the derived specialization within a delimited boundary (\emph{Inside-Out / Prompt}).
\end{enumerate}

\begin{table}[ht]
\centering
\small
\begin{tabular}{p{3.0cm} p{5.8cm} p{6.4cm}}
\toprule
\textbf{Direction / Polarity} & \textbf{Control Dimension} & \textbf{Data Dimension} \\
\midrule
\textbf{Shift / Outside-In}\newline (\emph{Derived wraps Base}) & \texttt{proceed()} / \texttt{super()}\newline (Specialized method wraps base continuation) & Subclass Record Extension ($S = \text{Base} \times \Delta$)\newline (Subclass ctor wraps base constructor) \\
\midrule
\textbf{Prompt / Inside-Out}\newline (\emph{Base encloses Derived}) & \texttt{inner()} (BETA Method Pattern)\newline (Base method encloses specialized hole $\mathcal{E}[\bullet]$) & Data Envelope / Virtual Pattern ($\mathcal{D}[\bullet]$)\newline (Base structure encloses specialized payload) \\
\bottomrule
\end{tabular}
\caption{The Polarized Duality of Specialization across Control and Data.}
\label{tab:specialization-duality}
\end{table}

In standard object-oriented languages (Simula, Smalltalk, Java), execution follows the \textbf{Shift / Outside-In} paradigm: dispatch begins at the most specialized subclass, which explicitly decides whether and when to delegate upward via \texttt{super()}. In AOP, AspectJ's \texttt{proceed()} adopts this exact outside-in direction: around advice intercepts execution and chooses whether to invoke the captured base continuation. However, this asymmetry creates the well-known \emph{fragile base class} problem: if a subclass or aspect omits \texttt{super()} or \texttt{proceed()}, base invariants and state updates are silently bypassed.

Conversely, the Scandinavian school of object-orientation pioneered by Ole Lehrmann Madsen in BETA~\cite{madsen1993beta} established the \textbf{Prompt / Inside-Out} paradigm. In BETA, execution begins at the outermost superclass; the superclass defines the invariant protocol and explicitly calls \texttt{inner()} to allow derived specializations to execute inside a delimited control hole $\mathcal{E}[\bullet]$. 

Crucially, this duality extends symmetrically to the data dimension:
\begin{itemize}
    \item \textbf{Data Shift (Subclass Extension):} In standard OOP, a subclass extends a superclass by concatenating fields ($S = \text{Base} \times \Delta$). The subclass acts as the outer container that inherits and owns the base components.
    \item \textbf{Data Prompt (Envelopes and Virtual Patterns):} The data counterpart to \texttt{inner()} is the \emph{Data Envelope} or \emph{Virtual Pattern}~\cite{madsen1993beta}. Here, the superclass defines an outer structural frame $\mathcal{D}[\bullet]$ (managing capabilities, arena boundaries, or object headers) that encloses an existential payload hole $\bullet$ instantiated by the derived class. In data construction, the base constructor initializes the outer envelope and calls an inverted virtual allocator hook to instantiate the interior payload within the established invariants.
\end{itemize}

\subsection{Spectators, Assistants, and the Aspect Liskov Principle}
Barbara Liskov and Jeannette Wing~\cite{liskov1994behavioral} formulated the Liskov Substitution Principle (LSP) for data objects: a subtype $S <: T$ must preserve all behavioral properties provable about $T$. In aspect-oriented programming, advice is non-oblivious and intercepts execution without caller cooperation. 

We formalize the \textbf{Aspect Liskov Substitution Principle} by categorizing advice into \emph{Spectators} and \emph{Assistants}~\cite{clifton2002spectators,dantas2006harmless}:

\begin{definition}[Aspect Liskov Substitution Principle]
Let joinpoint $J$ specify precondition $\text{Pre}_J(x)$ and postcondition $\text{Post}_J(x, y)$ under delimited answer type $\div C^+$. An advising aspect $A = (\text{adv}, \text{proceed})$ is a \emph{behaviorally sound substitute} for $J$ if:
\begin{enumerate}
    \item \textbf{Precondition Safety:} The advice precondition is weaker than the base joinpoint:
    \[
    \forall x.\; \text{Pre}_J(x) \implies \text{Pre}_{\text{adv}}(x)
    \]
    \item \textbf{Proceed Invariant:} Any invocation of $\text{proceed}(x')$ satisfies the target joinpoint's precondition:
    \[
    \text{Pre}_{\text{adv}}(x) \implies \text{Pre}_J(x')
    \]
    \item \textbf{Postcondition Safety:} The final result $y'$ delivered to caller continuation $K$ satisfies the joinpoint contract:
    \[
    \forall x, y'.\; \text{Post}_{\text{adv}}(x, y') \implies \text{Post}_J(x, y')
    \]
    \item \textbf{Answer-Type Containment:} The advice executes within the ambient prompt boundary without unhandled control escapes:
    \[
    \Gamma \mid \Delta \vdash A : (A^+ \to \sim B^+) \div C^+ \longrightarrow (A^+ \to \sim B^+) \div C^+
    \]
\end{enumerate}
\end{definition}

\paragraph{Pure Spectators (Harmless Advice).}
A pure aspect is a spectator: it observes cuts (e.g., auditing, contract verification) without mutating state or diverting control. In Polarized Scheme, spectators are typed as identity continuation transformers:
\[
T_{\text{spec}} : \sim B^+ \to \sim B^+ \quad (\text{with } \div C^+)
\]
They must invoke \texttt{proceed} exactly once with unmodified argument $x$ and return the unmodified result $y$, satisfying harmless non-interference ($\sigma \equiv_{\text{obs}} \sigma'$)~\cite{dantas2006harmless}.

\paragraph{Assistants (Around Advice via Delimited Control).}
An assistant aspect intercepts joinpoint execution and captures the delimited continuation $K$ awaiting the method's result. In Polarized Scheme, around advice is encoded directly using \texttt{with-continuation}:
\begin{lstlisting}
(define (around-aspect [target : (thunk (-> (pair int int) (~ int)))] ~k) : (thunk (-> int (~ int)))
  < (lambda ([factor : int] ~aret) : int
      < (pair 2 factor) | (with-continuation target ~aret) >) | ~k >)

(define advised-scale < c1 | (with-continuation around-aspect _HALT_) >)
< 4 | (with-continuation advised-scale _HALT_) >  ; yields 20
\end{lstlisting}
The around aspect receives the target object, constructs an interposed layered continuation that validates or transforms the scaling factor, invokes the captured method continuation (\texttt{proceed}), and returns the resulting value to \texttt{~aret}.

\paragraph{Bracket Semantics, Destructors, and Continuation Linearity.}
Just as object allocation couples constructors and destructors (RAII), around advice represents a control-level \textbf{Bracket Pattern} ($\text{Acquire} \to \text{Proceed} \to \text{Release}$). In Polarized Scheme, the lifecycle of an aspectual continuation layer is governed by the \emph{multiplicity} of the captured continuation co-variable:
\begin{itemize}
    \item \textbf{Linear Advice ($= 1$ proceed):} Strict RAII. The advice invokes \texttt{proceed} exactly once. The aspect constructor (allocating ephemeral continuation marks) and destructor (releasing resources and popping marks) execute in exact mirror order around the joinpoint.
    \item \textbf{Affine Advice ($\le 1$ proceed):} When around advice skips \texttt{proceed} (e.g., serving a cached value), the joinpoint is aborted. Because the Return Stack slice is abandoned, Polarized Scheme executes Return Stack unwind handlers to finalize pending mark scopes, preventing dangling resource leaks.
    \item \textbf{Unrestricted Advice ($*$, multiple proceeds):} If advice resumes a continuation multiple times (e.g., speculative backtracking or retries), the Return Stack slice is duplicated. To prevent double-free hazards, the type system requires that any continuation duplicated across multiple \texttt{proceed} calls must not close over non-copyable linear resources.
\end{itemize}

\subsection{Empirical Verification in Reference Evaluator}
We verified this case study across our reference implementation pipeline in VeGo. The macro expansion phase ($\mathcal{E} \vdash s \leadsto c$) lowers both object constructor macros (\texttt{define-class}) and aspect interposition wrappers (\texttt{define-aspect}) into pure $\lambda_{\text{pol}}$ cuts. The resulting terms pass bidirectional type checking and execute on the dual-stack machine, confirming that Polarized Scheme unifies object dispatch, delimited control, and aspectual continuation layers in a single orthogonal calculus. The complete reference evaluator, test suite, and macro pipeline are provided as ancillary files (\texttt{anc/}) accompanying this preprint.

\section{Conclusion and Future Work}
Polarized Scheme provides a disciplined foundation for Scheme by embracing the duality of computation. By separating positive data producers from negative observation consumers, functions express both value-transformer and continuation-transformer semantics symmetrically. Delimited control and continuation marks operate directly on the Return Stack parameterized by ambient answer types, and continuation-passing macros interleave cleanly with runtime execution. Future work includes expanding local bidirectional type inference to higher-rank polymorphic types and compiling polarized cuts directly to a register-allocated bytecode machine.

\bibliographystyle{plain}
\bibliography{references}

\end{document}